\documentclass[runningheads,a4paper]{llncs}
\usepackage[margin=1.25in]{geometry}
\usepackage{amssymb}
\usepackage{graphicx}
\usepackage{verbatim}
\usepackage{mathrsfs}
\usepackage{amsfonts}

\usepackage{amsthm}
\usepackage{ifpdf}
\usepackage{amsmath}
\usepackage[retainorgcmds]{IEEEtrantools}
\newtheorem{mydef}{Definition}
\newtheorem{mythm}{Theorem}
\newtheorem{mylmm}{Lemma}
\newtheorem{mycor}{Corollary}
\usepackage[lined,commentsnumbered,vlined]{algorithm2e}
\usepackage[noend]{algpseudocode}
\usepackage{enumerate}
\usepackage{url}

\urldef{\mailsa}\path|{ibanerje, richards}@cs.gmu.edu|    
\newcommand{\keywords}[1]{\par\addvspace\baselineskip
\noindent\keywordname\enspace\ignorespaces#1}

\begin{document}

\mainmatter  % start of an individual contribution

\title{Computing Maximal Layers Of Points in $E^{f(n)}$}

\titlerunning{Maximal Layers Problem}

\author{Indranil Banerjee, Dana Richards}
\authorrunning{On Maximal Layers Problem}

\institute{George Mason University\\ Department Of Computer Science\\ Fairfax Virginia 22030, USA\\
\mailsa}


\toctitle{Lecture Notes in Computer Science}
\tocauthor{Authors' Instructions}
\maketitle

\begin{abstract}
In this paper we present a randomized algorithm for computing the collection of maximal layers for a point set in $E^{k}$ ($k = f(n)$). The input to our algorithm is a point set $P = \{p_1,...,p_n\}$ with $p_i \in E^{k}$. The proposed algorithm achieves a runtime of $O\left(kn^{2 - {1 \over \log{k}} + \log_k{\left(1 + {2 \over {k+1}}\right)}}\log{n}\right)$ when $P$ is a random order and a runtime of $O(k^2 n^{3/2 + (\log_{k}{(k-1)})/2}\log{n})$ for an arbitrary $P$. Both bounds hold in expectation. Additionally, the run time is bounded by $O(kn^2)$ in the worst case. This is the first non-trivial algorithm 
whose run-time remains polynomial whenever $f(n)$ is bounded by some polynomial in $n$ while remaining sub-quadratic in $n$ for constant $k$. The algorithm is implemented using a new data-structure for storing and answering dominance queries over the set of incomparable points. 
\end{abstract}

\section{Introduction}

The problem of finding the maximal layers of a set $P = \{p_1,...,p_n\}$ of $n$ points\footnote{We have restricted the sampling set to $[0,1]^k$  in order to simplify our analysis. the results hold for any arbitrary compact subset of $E^k$.} in $[0,1]^k$ (where $k = f(n)$) is analogous to the problem of finding the convex layers of $P$. Given $P$ its first maximal  layer is defined to be the set $\mathcal{M}_1$ of points $q \in P$ such that for any other $p \in P$, $p \not \succ q$. Here, $\succ$ is any ordering relation between two points. For example, we could define $\succ$ as follow: $p \succ q$ if $p[j] \geq q[j]$ (where $p[j]$ is the $j^{th}$ coordinate of $p$) for all $j$. The first maximal set $\mathcal{M}_1$, which we simply refer to as the maximal  set of $P$, has been well studied \cite{Jour4,Jour5,Jour6}. The $l^{th}$ maximal layer $\mathcal{M}_l$ is recursively defined as the first maximal layer of remainder of $P$ upon removing from $P$ all the elements of layers from 1 to $l-1$. Note that $\mathcal{M}_l$ could be empty. The maximal layers problem is to identify all the non-empty maximal layers of $P$ and report them. We shall denote this problem as \textsc{MaxLayers}$(P)$.

\subsubsection{Related Work.} We only have a tight bound for \textsc{MaxLayers}$(P)$ when $k \leq 3$, which is $\Theta(n \log {n})$ \cite{Jour6,comm3}. However, we do not have any improved lower bound when $k > 3$. For fixed $k > 3$ best known upper-bound is $O(n (\log{n})^{k - 1})$ \cite{Jour3}. Interestingly, the upper bound to find only the first maximal set is $O(n (\log {n})^{\max(1, k - 2)})$. Both of these bounds hold in the worst-case.  We see that, for fixed $k$, these algorithms can be regarded as almost optimal, as they only have a poly-logrithmic overhead over the theoretical lower bound. Conceptually, they implement multi-dimensional divide and conquer algorithms \cite{Jour7} on input $P$ which introduces the poly-log factor in their runtimes. The point set $P$  is partitioned into subsets based on ordering of points in some arbitrary dimension. Then the maximal sets are computed recursively and merged later.

 Things get interesting if the number of dimensions is not bounded by a constant. When $k = \Omega(\log{n})$, these poly-logarithmic upper-bounds above becomes quasi-polynomial (in $n$). However, there is a trivial algorithm (which compares each point against the other, and keeps track of the computed transitive relations) that requires in the worst case $O(kn^2)$ comparisons. Although, for finding only the first maximal layer, \cite{comm1} proposed a deterministic algorithm that runs in $O(n^{(3 + \epsilon)/2})$ when $k = n$. Where, $O(n^\epsilon)$ is the complexity of multiplying two $n \times n$ matrices. So we see that the algorithm runs in $\omega(n^2)$ time. In a recent paper\cite{comm2}, authors show that for determining whether there exists a pair $(u, v)$, where $u \in A$ and $v \in B$ ($A$ and $B$ are both sets of vector of size $O(n)$) such that $u \succ v$, can be done in sub-quadratic time provided $k = O(\log{n})$. 
 
\subsubsection{Our Results.}
  In this paper we propose a randomized algorithm for the \textsc{MaxLayers}$(P)$ problem. When the point set $P$ is also a random order the runtime of our algorithm is bounded by $O\left(kn^{2 - {1 \over \log{k}} + \log_k{\left(1 + {2 \over {k+1}}\right)}}\log{n}\right)$ in expectation. Otherwise, it is $O(k ^2 n^{3/2 + (\log_{k}{(k-1)})/2}\log{n})$ also in expectation. Additionally, it takes $O(kn^2)$ time in the worst case. This is the first non-trivial algorithm for which the following two conditions holds simultaneously: 1) The worst case run time is polynomially bounded (in $n$) as long as $k$ is bounded by some polynomial in $n$. 2) Whenever $k$ is a constant the run time of the proposed algorithm is sub-quadratic in $n$ (in expectation).
%------------------------------------------------------------------------------
%
%
%
%------------------------------------------------------------------------------
\section{Preliminaries}
We denote $P = \{p_1,...,p_n\}$ as the input set of $n$ points in $E^k$. The $j^{th}$ coordinate of a point $p$ is denoted as $p[j]$. For any points $p, q \in P$, we define an ordering relation $\succ$, such that $p \succ q$ if $p[j]\ op\ q[j]\ \forall j \in [1...k]$. Where, $op$ is a place holder for $\geq$ or $\leq$. Consequently, there are $2^k$ different ordering relations $\succ$ and for each such an ordering there is a unique set of maximal layers (of $P$). Without loss of generality, we  assume that $op $ is $ \ge$ for all $j$, in this paper. Henceforth we will simply use $\ge$ in place of $op$. We will use the notation $S \succ p$, where $S$ is a set of incomparable elements, to denote that $\exists q \in S$ such that $q \succ p$. If $S \succ p$ we say that $S$ is ``above'' $p$. Furthermore, if $p \succeq q$ then either $p = q$ or $p \succ q$.

Clearly, $(P, \succ)$ defines a partial order. We shall simply use $P$ to denote this poset when the context is clear. If $p \succ q$ then we say that $p$ precedes (or dominates) $q$ in the partial order and that they are comparable. We say that $p$ and $q$ are incomparable (denoted by $p \parallel q$) if $p \not \succ q$  and $q \not \succ p$. If $p$ and $q$ belong to the same maximal layer then $p \parallel q$. Let the height $h$ of $P$ be defined as the number of non-empty maximal layers of $P$. We also define the width $w$ of $P$ as the size of the largest subset of $P$ of mutually incomparable elements. Note, that the maximum size of any layer is $\leq w$.

Let $\mathcal{O} : E^k \times E^k \to \{0,1\}^k$, such that $\operatorname{\mathcal{O}}(p,q)[j] = 1$ if  $p[j] < q[j]$ and 0 otherwise. This  definition, which might seem inverted, will make sense when we discuss it in the context of our data structures. We call $\mathcal{O}$ the orthant function as it computes the orthant with origin $p$ in which $q$ resides. Henceforth, the maximal layers will simply be referred to as layers. Let $T$ be a linear ordering of the points in $P$ such that for any $p,q \in P$, if $p \succ q$ then $p$ precedes $q$ in $T$, that is, a liner-extension preserves the precedence relations between the elements of $P$. Let $|S|$ denote the size of the set $S$.
%------------------------------------------------------------------------------
We are now ready to state the  \textsc{MaxLayers}$(P)$ problem formally:
\begin{mydef}[\textsc{MaxLayers}$(P)$]
Given a point set $P$ along with an ordering relation defined above, label each point in $P$ with rank of the maximal layer it belongs to.
\end{mydef}
%------------------------------------------------------------------------------
In our analysis we shall use the typical RAM model, where operation of the form $p[j] \ge q[j]?$  takes constant time. In our analysis we shall first assume that the point set $P$ forms a random order. Then we will extend the result for an arbitrary set of points. Below we define random orders formally according to its definition in \cite{Jour1}.

\begin{mydef}
   We pick a set of $n$ points uniformly at random from $[0, 1]^k$. Then the partial order generated by these points is a {\it random order}.
\end{mydef}

This is equivalent to saying that $(P, \succ)$ is the intersection of $k$ linear orders $T_1 \times ... \times T_k$ where the $k$-tuple $(T_1,...,T_k)$ is chosen uniformly at random from $(n!)^k$ such tuples. Here, each $T_j$ is a linear ordering (permutation) of $\{1,2,...,n\}$. Whenever we present our run-time results in terms of $w$ or (and) $h$ it is assumed that both are upper bounded by $n$, the number of points in $P$. To simplify our analysis we ignore the expected values of $w$ and $h$, which could only have made our results stronger (for example, see \cite{Jour1,Jour2}).
%------------------------------------------------------------------------------
%
%
%
%------------------------------------------------------------------------------
\section{The Iterative Algorithm}
We shall use \textsc{MaxPartition}$(P)$ as the main procedure for solving an instance of \textsc{MaxLayers}$(P)$. First we will describe a simpler algorithm and analyze it for a random order $P$. Then we extend it for an arbitrary set of points.
%------------------------------------------------------------------------------
\subsection{Data Structures}
In this section we introduce the framework on which  our algorithm is based. Let $B$ be a  self-balancing binary search tree (for example $B$ could be realized as a red-black tree). Let $B(i)$ be the $i^{th}$ node in the in-order of $B$. Each node of $B$ stores three pointers. One for each of its children (null in place of an empty child) and another pointer which points to an auxiliary data structure. If $X$ is a node in $B$ then left and right children of $X$ are denoted as $l(X)$ and  $r(X)$ respectively. We also denote by $L(X)$ the auxiliary data structure associate with $X$. When the context is clear,  we shall simply use $L$ in place of $L(X)$.

We also let $L$ be a placeholder for any data structure that can be used to store the set of points from a single layer of $P$. For example, $L$ could be realized as a linked list. Additionally, $L$ must support \textsc{Insert}$(L, p)$ and \textsc{Above}$(L, p)$. The \textsc{Above}$(L, p)$ operation takes a query point $p$, and answers the query $L \succ p ?$ . The \textsc{Insert}$(L, p)$ operation  inserts $p$ into $L$, which assumes $p$ is incomparable to the elements in $L$. So, we must ensure that $L$ is the correct layer for $p$ before calling \textsc{Insert}$(L, p)$.

We observe that the layers of $P$ are themselves linearly ordered by their ranks from 1 to $h$. We can thus use $B$ to store the layers in sorted order, where each node $B(i)$ would store the corresponding layer $\mathcal{M}_i$ (using $L(B(i))$). We endow $B$ with \textsc{Insert}$(B, p)$ and \textsc{Search}$(B, p)$ (we do not need deletion) operations. The \textsc{Insert}$(B, p)$ procedure first calls the \textsc{Search}$(B, p)$ procedure to identify which node $B(i)$ of $B$  the new point $p$ should belong and then calls \textsc{Insert}$(L(B(i)), p)$. If $p$ does not belong to any layer currently in $B$ then we create a new node in $B$. The \textsc{Search}$(B, p)$ procedure works as follows: we can think of $B$ as a normal binary search tree, where the usual comparison operator $\ge$ has been replaced by the \textsc{Above}$(L, p)$ procedure. Furthermore, the procedure can only identify whether $L \succ p$ or $L \not \succ p$. This is exactly equivalent to the situation where we have replaced the comparison operator $\ge$ with $>$. So we must determine two successive nodes $B(i)$ and $B(i+1)$ such that $L(B(i)) \succ p$ and  $L(B(i + 1)) \not \succ p$. If such a pair of nodes does not exist then we return a null node.
 %------------------------------------------------------------------------------
\subsection{\textsc{MaxPartition}$(P)$}
We begin by first computing a linear extension $T$ of $P$. We initialize $B$ as an empty tree. We iteratively pick points from $P$ in increasing order of their ranks in $T$ and call \textsc{Insert}$(B, p)$, where $p$ is the current point to be processed. \textsc{Insert}$(B, p)$ subsequently calls \textsc{Search}$(B, p)$. We have two possibilities: 

\textsc{case 1:} \textsc{Search}$(B, p)$ returns a non-empty node $B(i )$. We then call \textsc{Insert}$(L(B(i), p)$.

\textsc{case 2:} \textsc{Search}$(B, p)$ returns a null node. Then we create the node $B(m + 1)$ in $B$, where $m$ is the number of nodes currently in $B$. We first initialize $B(m + 1)$ and then call \textsc{Insert}$(L(B(m + 1), p)$ on it. We note that, when we create a new node in $B$ it must always be the right-most node in in-order of $B$. This follows from the order in which we process the points. Since $p$ succeeds a processed point $q$ in the linear extension $T$, hence $p \not \succ q$. Thus, if $p$ does not belong to any of nodes currently  in $B$ then it must be the case that $p$ is below all layers in $B$.

%------------------------------------------------------------------------------
\textsc{MaxPartition}$(P)$ terminates after all points have been processed. At termination $L(B(i))$ stores all of the points in $\mathcal{M}_i$ for $1 \le i \le h$. We make a couple of observations here. 1) When a point is inserted into a node $B(i)$ it will never be displaced from it by any point arriving after it.  2) Since, nodes are always added as the right-most node in $B$, for  \textsc{Search}$(B, p)$ to be efficient, $B$ must support self-rebalancing.

 If we assume that \textsc{Above}$(L, p)$ and \textsc{Insert}$(L, p)$ to work correctly, at once we see that \textsc{Search}$(B, p)$ and \textsc{Insert}$(B, p)$  are also correct. Hence, each point is correctly assigned  to the layer it belongs to.   

%------------------------------------------------------------------------------
\subsection{Runtime Analysis}
Let \textsc{Above}$(L, p)$ take $t_a(|L|)$ time. As mentioned in section 2, $|L| \leq w$ for any layer in $B$. Hence, $t_a(w)$ is an upper bound on the runtime of \textsc{Above}$(L, p)$. Similarly, we bound the  runtime of \textsc{Insert}$(L, p)$ with $t_{i}(w)$. Let $p$ be the next point to be processed. At the time $B$ will have at most $h$ nodes. In order to process $p$ the  \textsc{Insert}$(B, p)$ will be invoked, which in turn calls the  \textsc{Search}$(B, p)$ as discussed above. But the  \textsc{Search}$(B, p)$  will employ a normal binary search on $B$ with the exception that at each node of $B$ it invokes the \textsc{Above}$(L, p)$ instead of doing a standard comparison. Since, $B$ is self-balancing the height of $B$ is bounded by $O(\log{h})$. Hence, number of calls to \textsc{Above}$(L, p)$ is also bounded by $O(\log{h})$, each of which takes $t_a(w)$ time. Also, for each point $p$, \textsc{Insert}$(L, p)$ is called only once. We also assume initializing a node in $B$ takes constant time. So, processing of $p$ takes $O(t_a(w) \log{h} + t_{i}(w))$ and this holds for any point. 
%------------------------------------------------------------------------------
\begin{mylmm}
   We can compute a linear extension $T$ of $P$ in $O(n \log{n} + kn)$ time in the worst case.
\end{mylmm}
\begin{proof}

We shall compute $T$ as follows: Let $\mu(p) = \max_{1 \le j \le k}\ {p[j]}$.  Then sorting the points in  decreasing order of $\mu(p)$ will give us $T$. It is trivial to see that $T$ is a linear extension of $P$. This takes $O(n \log{n} + kn)$ in the worst case. 

\end{proof}
%------------------------------------------------------------------------------

%------------------------------------------------------------------------------
The reason for computing $T$ in this way will be clear when we get to the analysis of our algorithm. Later we shall see that the time bounds for \textsc{Above}$(L, p)$ and \textsc{Insert}$(L, p)$ will dominate the time it takes to compute $T$. So we shall ignore this term in our run-time analysis. The next theorem trivially follows from the discussion above.
\begin{mythm} The procedure \textsc{MaxPartition}$(P)$ takes $O(n (t_a(w)\log{h} + t_{i}(w)))$ time and upon termination outputs a data structure consisting of the maximal layers of $P$ in sorted order. 
\end{mythm}

%------------------------------------------------------------------------------
%Here the expectation on the runtime is taken over both the internal randomness of $L$ and  of $P$. 
%---------------------------------------------- --------------------------------
%
%
%
%------------------------------------------------------------------------------
 
%------------------------------------------------------------------------------
%
%
%
%------------------------------------------------------------------------------
\section{Realization of $L$ using Half-Space Trees}
In this section we introduce a new data structure for implementing $L$. We shall refer to it as Half-Space Tree (HST). 

The function $\operatorname{\mathcal{O}}(p,q)$ computes which orthant $q$ belongs to with respect to $p$ as the origin. Clearly, there are $2^k$ such orthants, each having a unique label in $\{0,1\}^k$. Let $H_j(p)$ be a half space defined as: $H_j(p) = \{q \in [0,1]^k\ | \operatorname{\mathcal{O}}(p,q) = \{0,1\}^{j-1}0\{0,1\}^{k-j}\}$ passing through origin $p$ whose normal is parallel to dimension $j$. Here, $\{0,1\}^{j-1}0\{0,1\}^{k-j}$ represents a 0-1 vector for which the $j^{th}$ component is 0. We shall use the notation $h_j(p)$ to denote the extremum orthant of $H_j(p)$ (w.r.t $\succ$), that is, $h_j(p) = 1^{j-1}01^{k-j}$. There are $k$ such half spaces. An orthant whose label contains $m$ 1's lies in the intersection of some $k - m$ such half spaces. 
%------------------------------------------------------------------------------
\begin{mylmm}
   If $p,q \in P$ and $p \parallel q$ then $\operatorname{\mathcal{O}}(p,q) \in \{0,1\}^k\setminus \{0^k,1^k\}$. That is, $q$ can only belong to orthants which lie in the intersection of at most $k-1$ half spaces. 
\end{mylmm}
\begin{proof}
   Trivially follows from definitions. 
\end{proof}
\begin{mycor}
   The above lemma holds if $p$ and $q$ belongs to the same layer. However, the converse of this statement is not true.
\end{mycor}
%------------------------------------------------------------------------------
\subsection{Half-Space Tree}

We define a $k$-dimensional HST recursively as follows:
\begin{mydef}[HST]   
\begin{enumerate}
   \item A singleton node (root) storing a point $p$.
   \item A root has a number of non-empty children nodes (up to k) each of which is a HST.
   \item If node $q$ is the $j^{th}$ child of node $p$ then $h_j(p) \succeq \operatorname{\mathcal{O}}(p,q)$.
\end{enumerate}  
\end{mydef} 
%------------------------------------------------------------------------------
An HST stores points from a single layer. So Corollary 1 tells us that for any  node $p$ and a new point $q$ at most $k- 1$ of the children nodes satisfy  $h_j(p) \succeq \operatorname{\mathcal{O}}(p,q)$. Hence, $q$ can be inserted into any one out of these children nodes. Henceforth, we will also use $w$ (the width of $(P, \succ)$) to bound the number of points currently stored inside $L$.
%------------------------------------------------------------------------------
\subsubsection{\textsc{Above}$(L, p)$}
   Let us assume that $L$ is realized by an HST. The \textsc{Above}$(L, p)$ works as follows: First we compute $\operatorname{\mathcal{O}}(r,p)$. Here, $r$ is the root node. If $\operatorname{\mathcal{O}}(r,p) = 0^k$ then we return $L \succ p$. Otherwise we call \textsc{Above}$(j(L), p)$ recursively on each non-empty child node $j$ of root $r$, such that  $h_j(r) \succeq \operatorname{\mathcal{O}}(r,p)$. When all calls reach some leaf node, we stop and return $L \not \succ p$.

\begin{proof}[of correctness]
   \textsc{case 1:}($L \succ p$) Let $q$ be some point in $L$ such that $q \succ p$, prior to calling  \textsc{Above}$(L, p)$. Before reaching the node $q$, if we find some other node $q' \succ p$ then we are done. So we assume this is not the case. We claim that $p$ will be compared with $q$.  We show this as follows: Let the length of path from root $r$ to $q$ be $i+1$. Let $u_0, ..., u_{i}$ be the sequence of nodes in this path (here $u_0 = r$ and $u_i = q$). Since, $q \succ p$, $\operatorname{\mathcal{O}}(u_m,q) \succeq \operatorname{\mathcal{O}}(u_m,p)$ for all $0 \leq m < i$. But, $u_m$ is a predecessor node in the path from $r$ to $q$, hence $h_{j_{m}}(u_m) \succeq \operatorname{\mathcal{O}}(u_{m},q)$ where $u_{m+1}$ is the $j_m^{th}$ child of $u_m$. Which implies $h_{j_{m}}(u_m) \succeq \operatorname{\mathcal{O}}(u_m,p)$ (from transitivity of $\succeq$) for $0 \leq m < i$. Thus we will traverse this path at some point during our search.
      
   \textsc{case 2:}($L \not \succ p$) Follows trivially from the description of \textsc{Above}$(L, p)$.
\end{proof}
%------------------------------------------------------------------------------
\subsubsection{ \textsc{Insert}$(L, p)$}
\textsc{Insert}$(L, p)$ is called with the assumption that $L \not \succ p$. If the root is empty then we make $p$ as the root and stop. Otherwise, we pick one element uniformly at random from the set $S_r = \{j \in \{1,...,k\}\ |\ h_j(r) \succeq \operatorname{\mathcal{O}}(r,p)\}$ and recursively call \textsc{Insert}$(j(L), p)$.

\begin{proof}[of correctness]
   It is easy to verify that insert procedure maintains the properties of HST given in definition 2.
\end{proof}

Although the insert procedure is itself quite simple, it is important that we understand the random choices it makes before moving further. These observation will be crucial to our analysis later. Let the current height of $L$ be $h_L$. By $L^*$ we denote the complete HST of height $h_L$, clearly $L^*$ has $k^{h_l}$ nodes. We color edges of $L^*$  red if both of the nodes it is incident to are present in $L$, otherwise we color it blue. Unlike \textsc{Above}, we can imagine that the \textsc{Insert} procedure  works with $L^*$ instead of $L$. Upon reaching a node $r$ in $L^*$ the procedure samples uniformly at random from the set $S_r$ as above. This set may contain edges of either color. If a blue edge have been sampled then we stop and insert $p$ into the empty node incident to the blue edge in $L$. So we see that, despite not being in $L$, the nodes incident to blue edges effect the sampling probability equally.

%------------------------------------------------------------------------------
\subsection{Runtime Analysis}
Here we compute $t_a(w)$ and $t_{i}(w)$ in expectation over the random order $P$ and the internal randomness of the \textsc{Insert}$(L, p)$ procedure. From the discussion in section 4.1 we clearly see that  $t_{i}(w) = O(t_a(w))$. So it suffices to upper bound $t_a(w)$ in expectation. Furthermore, we only need to consider the case when \textsc{Above}$(L, p)$ returns $L \not \succ p$ as the other case would take fewer number of comparisons. Let this time be $u(w)$. We divide our derivations to compute $u(w)$ into two main steps:

\begin{enumerate}[i.]
   \item Compute the expected number of nodes at depth $d$ of $L$ having $w$ nodes. 
   \item Use that to put an upper bound on the number of nodes visited during a call to \textsc{Above}$(L, p)$ (when $L \not \succ p$).
\end{enumerate}

We choose to process points according to $T$ as detailed earlier. We denote this ordering by the ordered sequence $(p_1,...,p_n)$. 
\begin{mylmm}
   For any two points $p, q$ where $p$ precedes $q$ in $T$ we have the probability that $p[j] > q[j]$ is $\eta_1(k) = 1 - {1 \over 2}{{k - 1} \over {k + 1}}$. Additionally, if $p$ and $q$ are incomparable then it is $\eta_2(k) = 1 - {1 \over k} - {1 \over 2}{{k - 2} \over {k + 2}}$.  
\end{mylmm}

\begin{proof}
   See appendix.
\end{proof}

\begin{mythm}
   After $w$ insertions the expected number of nodes at depth $d$ in $L$ is given by: \[ k^d\left(1 - \sum_{i = 1}^{d}{{(1 - {1\over {k^i}})^{w - 1}}\over{\prod_{j = 1, j \ne i}^{d}{(1 - {1\over {k^{i - j}}})}}}\right)\] 
\end{mythm}
\begin{proof}
   Let $X_{w,d}$ be the number of nodes at depth $d$ of $L$ after $w$ insertions. Due to the second assertion of Lemma 3 we know that any new point to be inserted can belong to any of the $k$ half-spaces with probability $\eta_2(k)$, which is constant over the half-spaces. The insert procedure selects one of these candidate half-spaces uniformly at random. Thus it follows from symmetry that a particular half-space will be chosen for insertion with probability ${1 \over k}$. If the subtree is non-empty then we do these recursively. We define an indicator random variable for the event that the $t^{th}$ insertion adds a node at depth $d$ as $I_{t,d}$. Then, \[X_{w,d} = \sum_{t = 1}^{w}{I_{t,d}}\] Taking expectation on both side we get,
    \[\mathbb{E}[X_{w,d}] = \sum_{t = 1}^{w}{\operatorname{Pr}[I_{t,d}]}\]
     Trivially, $\mathbb{E}[X_{w,0}] = 1$ for $t > 0$. When $d = 1$ and $t \ge 2$ then $\operatorname{Pr}{[I_{t,1}]} = 1 - {X_{t-1,1}\over k}$. This is because there are $X_{t - 1,1}$ nodes at depth 1 (nodes directly connected to the root) hence there are $k - X_{t - 1,1}$ empty slots for the node to get inserted at depth 1, otherwise it will be recursively  inserted to some deeper node. Hence we have, 
    \[\mathbb{E}[X_{w,1}] =  \sum_{t = 2}^{w}{\left(1 - {X_{t-1,1}\over k}\right)}\] 
    For $d = 2$, we can similarly argue that the probability of insertion at depth 2 for some $t \ge 3$ is equal to probability of reaching a node at depth 1 times the probability of being inserted at depth 2. It is not difficult to see that this equals: $\left({X_{t-1,1} \over k}\right)\left(1 - {X_{t-1,2} \over k X_{t-1,1} }\right)$. Hence, \[\mathbb{E}[X_{w,2}] =  \sum_{t = 3}^{w}{\left({X_{t-1,1} \over k}\right)\left(1 - {X_{t-1,2} \over k X_{t-1,1} }\right)}\] Proceeding in this way we see that,
     \[\mathbb{E}[X_{w,d}] =  \sum_{t = 1}^{w}{\left({\mathbb{E}[X_{t-1,d-1}] \over {k^{d-1}}} - {\mathbb{E}[X_{t-1,d}] \over {k^{d}}}\right)}\]
     Here we again take expectation on both sides and simplify the expression so that the sum starts from $t = 1$ since the terms $\mathbb{E}[X_{t,d}] = 0$ when $t \le d$.
   
   Let $a(w, d) = \mathbb{E}[X_{w,d}]$,  we can then simplify the above equation to get the following recurrence, 
   \[a(w, d) = {a(w-1, d-1) \over k^{d-1}} + \left(1 - {1 \over k^d}\right)a(w - 1, d)\] 
   with $ a(w, d) = 0 $ for $w \le d$. The solution to this can be found by choosing a ordinary generating function $G_d(z)$ with parameter $d$, such that $G_d(z) = \sum_{t = 0}^{\infty}{a(t,d)z^t}$. The solution [see appendix] completes the proof of the theorem. 
 
\end{proof}

Before moving on to the main theorem we need another lemma:
\begin{mylmm}
      If $ B = (b_0, b_1, ..., b_n)$ is a sequence such that $b_{r} \ge b_{r+ 1} \ge ... \ge b_{n}$, then the sum $S = \sum_{i = 0}^{n}{b_im^i} \le \sum_{i = 0}^{r}{b_im^i} + {b_{r+1}m^{r+1} \over (1 - m)}$ where $m < 1$. 
\end{mylmm}

\begin{proof}
   See appendix. 
\end{proof}

\begin{mycor}
   If $m = 1 - {1 \over 2}{{k - 1} \over {k + 1}}$ and $k \ge 4$ then , $S \le \sum_{i = 0}^{r}{b_im^i} + {7 \over 3}b_{r+1}m^r$.
\end{mycor}

\begin{mythm}
   Expected number of nodes visited during an unsuccessful search $u(w)$ is bounded by $O\left(w^{1 - {1\over{\log{k}}} + {\log_{k}{\left(1 + {2 \over {k+1}}\right)}}}\right)$.
\end{mythm}
\begin{proof}
   Before proving this we make the following observation. If for any $d = d_0$, the sequence $a(w, d)$ becomes decreasing, that is, $a(w, d_0) \ge a(w, d_0 - 1)$ and $a(w, d_0) > a(w, d_0 + 1)$, then afterwards it will stay decreasing. This is clear from the fact that $a(w, d)$ represents the expected number of nodes at depth $d$ after $w$ insertions. So the sequence $a(w, d)$ is unimodal since $a(w, 0) \le a(w, 1)$ trivially for $w \ge 2$. Let $d_0$ be the value that maximizes $a(w, d)$.

   Let us compute the probability of visiting a node at depth $d$ during a call to \textsc{Above} when the query point is not below $L$. Let $q$ be the current node being checked and $p$ be the query point. 
According to Lemma 3  the probability $\operatorname{Pr}[p \in H_j(q)]$ is same for any $j$ and is not dependent on the rank of $q$ in $T$. Hence it is also not dependent on the depth of $q$ in $L$. Furthermore, this probability  is $\eta_1 = 1 - {1 \over 2}{{k - 1} \over {k + 1}}$, again from Lemma 3.

 Thus the probability of visiting a node at depth $d$ is the result of $d$ independent moves each having probability $\eta_1$, hence it is $\eta_1^d$. Now we can find the expression for the expected number of nodes visited:
\begin{align}
      u(w) &= \sum_{d = 0}^{w - 1}{\eta_1^d a(w, d)} \nonumber \\
    &\le \sum_{d = 0}^{d_0}{\eta_1^d a(w, d)} + {7 \over 3}\eta_1^{d_0}a(w,d_0 + 1)\nonumber \\
\end{align}
Here we use Theorem 2, Lemma 4 and its corollary and the fact that the sequence $a(w, d)$ is unimodal; to bound $u(w)$. Also note that $a(w,d) \le k^d$. Now we need to upper bound $d_0$. With some tedious algebra [see appendix] we get, $d_0 \le \log_{k}{w} + 2$. Again, after some more algebra [see appendix] we finally get,
    
\begin{eqnarray}
   u(w) \le  O\left(w^{1 - {1 \over {\log{k}}} + {\log_{k}{\left(1 + {2 \over {k+1}}\right)}}}\right)
\end{eqnarray}

   This proves Theorem 3. 
\end{proof}
\begin{mycor}
   The algorithm runs in $O\left(kn^{2 - {1 \over \log{k}} + \log_k{\left(1 + {2 \over {k+1}}\right)}}\log{n}\right)$ in expectation. 
\end{mycor}
\begin{proof}
   From Theorem 1 and the first paragraph of Section 4.2 we see that the runtime of 
   the algorithm is $O(knu(w)\log{h})$. Since computing $\operatorname{\mathcal{O}}(p, q)$ between pairs of vectors takes $O(k)$ time. Using the upper-bound of $u(w)$ and the fact that $w, h \le n$ we get the runtime as claimed above.
   
\end{proof}

\section{Extension to Arbitrary $P$}
The previous algorithm would still be correct if $P$ is not a random order. However the expected runtime will no longer hold. In order make our previous analysis work for any set of points we modify the way we store the layers. In this new setting layers are still arranged using a balanced binary $B$, exactly as before. However, each layer is now stored using a list of HSTs instead of just a single one. Let us call this data structure \textsc{List-HST}. We extend the \textsc{Above} and \textsc{Insert} procedure for HST in a obvious way. 

\textsc{List-HST} starts with an empty list. Attached to a \textsc{List-HST} is another list $R$ in which newly arrived points are kept temporarily before they are ready to be inserted in the \textsc{List-HST}. Initially this list $R$ is also empty. We take the maximum size of $R$ as $\sqrt{w}$. As long as $R$ has less than $\sqrt{w}$ points we keep adding to it. Once $R$ has been filled, we create an HST from the points in $R$ and remove these points from R. This becomes the first HST in the list. We repeat these steps again when $R$ is full. We describe the modified \textsc{List-HST-Above} and \textsc{List-HST-Insert} below.

\subsubsection{\textsc{List-HST-Insert}($p$)} If $R$ is not yet full then we just add $p$ to $R$. Otherwise we create HST from points in $P' = R \cup p$. We randomly permute elements in $P'$ and pick the first element in this ordering as the root. We then build the HST iteratively by picking elements in this order. Next we prove a lemma similar to Lemma 3.

\begin{mylmm}
   Assume that an HST is build by inserting points in a random order. Let $X$ represents a point which is already inserted and $Y$ a new point being compared to $X$. Then $Y$ belongs to any of the children subtree of $X$ with equal probability.
\end{mylmm}

\begin{proof}
   Since the insertion order is random, $X$ and $Y$ are both random variables. We compute the probability $\operatorname{Pr}[X[i] > Y[i]]$ for some $i$. Now for some arbitrary pair of points $\{p, q\}$ the probability that $p$ precedes $q$ in the ordering is $1/2$. If $X, Y \in \{p, q\}$ then, $\operatorname{Pr}[X[i] > Y[i]| X, Y \in \{p,q\}] = (\operatorname{Pr}[X[i] > Y[i]| X=p, Y=q] + \operatorname{Pr}[X[i] > Y[i]| X=q, Y=p])/2$. $\operatorname{Pr}[p[i] > q[i]] $is either 0 or 1 since the points $p$ and $q$ themselves are not random. Hence, $\operatorname{Pr}[X[i] > Y[i]| X, Y \in \{p,q\}] = 1/2$ for any pair $\{p, q\}$, which proves the claim above.
\end{proof}

Using the above lemma and techniques used to prove Theorem 3 we can show that $d_0 \le \log_{k}{\sqrt{w}} + 2$. Hence, building an HST takes $O(\sqrt{w}\log_{k}{w})$ time in expectation and $O(w)$ in worst case. Since, there can only be $\sqrt{n}$ such steps where we build an HST and each takes $O(n)$ (since $w \le n$) time hence the insert operation on \textsc{List-HST} adds $O(n^{3/2})$ to the overall running time of our algorithm. This is insignificant compared to the total time. 

\subsubsection{\textsc{List-HST-Above}($p$)} For each HST $L$ in the list we call \textsc{Above($L,p$)}. If none of these calls find a point above $p$ then we check the remaining points in $R$. However, since $p$ is not random we cant compute the probability $\eta_1$ as we did before. However, we can upper bound the fraction of subtrees that are visited from a node. We see that a point $p$ can visit at most $k-1$ subtrees of a node $q$ otherwise we can conclude that $q \succ p$. The \textsc{List-HST-Insert} procedure creates the $j^{th}$ subtree of $q$ with equal probability for all $j$. Hence, during the search step the point $p$ will visit a non-empty subtree of $q$ is with probability $\le (k-1)/k$. This value can be substituted as an upper bound for $\eta_1$ in Equation 1, which leads to $u(w) \le O(k\sqrt{w}^{\log_{k}{(k-1)}})$. Since there are at most $\sqrt{w}$ HSTs in a layers, it takes $O(kw^{1/2 + (\log_{k}{(k-1)})/2})$ time in expectation to search a list of HSTs. We ignore the time it takes to check the set $R$, which is $O(k\sqrt{w})$. Hence the total runtime is bounded by $O(k^2 n^{3/2 + (\log_{k}{(k-1)})/2}\log{n})$ in expectation.
\subsection{A Summary of Results}
We summarize the main results as follows:
\begin{enumerate}[i.]
   \item{\textit{$k$ is a constant.}} From Corollary 3 we can easily verify that the algorithm has a runtime of $O(n^{2 - \delta(k)})$  where $\delta(k) > 0$. This remains true even when $P$ is not a random order.
   \item{\textit{$k$ is some function of $n$}} We let $k = f(n)$. For any $k$ the runtime of our algorithm is bounded by $O(kn^2)$ in the worst case. This bound does not hold for the divide-and-conquer algorithm in \cite{Jour3}. Also, The proposed algorithm never admits a quasi-polynomial runtime unlike any of the previously proposed non-trivial algorithms.
\end{enumerate}

%------------------------------------------------------------------------------
%
%
%
%------------------------------------------------------------------------------
\subsection*{Concluding Remarks}
In this paper we proposed a randomized algorithm  for the \textsc{MaxLayers}$(P)$ problem.  Unlike previous authors we also consider the case when $k$ is not a constant; this is often the case for many real-world data sets whose tuple dimensions are not insignificant with respect to its set size. In this setting we show that the expected runtime of our algorithm is $O\left(kn^{2 - {1 \over \log{k}} + \log_k{\left(1 + {2 \over {k+1}}\right)}}\log{n}\right)$ when $P$ is a random order. For any arbitrary set of points in $E^k$ it exhibits a runtime of $O(k^2 n^{3/2 + (\log_{k}{(k-1)})/2}\log{n})$ in expectation. It remains to be seen if there exists a determinist algorithm that runs in $o(kn^2)$ for this problem. As a future work it would be interesting to know whether HST can be used for the unordered convex layers problem in higher dimensions. We know that unlike the maximal layers problem this problem is not decomposable  \cite{book1}. So it would be interesting to know within our iterative framework whether we can extend HST to store the convex layers also.

%------------------------------------------------------------------------------
%
%
%
%------------------------------------------------------------------------------

%------------------------------------------------------------------------------
%
%
%
%------------------------------------------------------------------------------
\section*{Appendix}
\subsection*{Proof of Lemma 3}
\begin{proof}
 Recall that $T$ is a linear extension of $P$. Since $p$ precedes $q$ in $T$, $\mu(p) > \mu(q)$. Hence, $\exists$ $j^{'} \in \{1,...,k\}$ such that $p[j'] > q[j']$. Let $j' = \operatorname{argmax}_{1 \le j \le  k}{p[j]}$. We compute the probability  $\operatorname{Pr}[p[j] > q[j]\ |\  \mu(p) > \mu(q)]$ in two parts over the disjoint sets $\{j = j^{'}\}$ and $\{j \ne j^{'}\}$:

\begin{align}
\nonumber   \operatorname{Pr}[p[j] > q[j]\ |\  \mu(p) > \mu(q)] &= \operatorname{Pr}[p[j] > q[j]\ |\ j = j^{'}, \mu(p) > \mu(q)]\operatorname{Pr}[j = j^{'}] \\
\nonumber&+ \operatorname{Pr}[p[j] > q[j]\ |\ j \not = j^{'},  \mu(p) > \mu(q)]\operatorname{Pr}[j \not = j^{'}]\\
   & = 1{1 \over k} + \left(1 - {1 \over 2}{\mu(q) \over \mu(p)}\right)\left(1 - {1 \over k}\right) 
\end{align}

Since,
   \begin{align}
\nonumber
\operatorname{Pr}[p[j] > q[j]\ |\ j \not = j^{'},  \mu(p) > \mu(q)]
   & = {{(\mu(p) - \mu(q))\mu(q) + {\mu(q)^2 \over 2}} \over  {\mu(p)\mu(q)}}=1 - {1 \over 2}{\mu(q) \over \mu(p)} 
\end{align}

\noindent This follows from the fact that $p[j]$ and $q[j]$ are independent random variables  uniformly distributed over $[0, \mu(p)]$ and $[0, \mu(q)]$ (given $\mu(p) > \mu(q)$) respectively. In the set $\{j = j^{'}\}$ clearly $p[j] > q[j]$. However, in the set $\{j \ne j^{'}\}$ the probability that $p[j] > q[j]$ is $\left(1 - {1 \over 2}{\mu(q) \over \mu(p)}\right)$. We note that $\mu(p), \mu(q)$ are themselves random variables. More importantly they are i.i.d random variables having the following distribution: \[\operatorname{Pr}[\mu(p) < t] = t^{k}\] on the interval $[0,1]$. This follows from how points in $P$ are constructed. We take the expectation of both side of Equation 3.1 over the event space generated by $\mu(p), \mu(q)$ on the set $\{\mu(p) > \mu(q)\}$:
  \begin{align}
     \nonumber
     \operatorname{Pr}[p[j] > q[j]\ |\  \mu(p) > \mu(q)] &= 1 - {1 \over 2}\mathbb{E}\left[{\mu(q) \over \mu(p)}\ \mid\ \mu(p) > \mu(q)\right]\left(1 - {1 \over k}\right) \\ \nonumber
     &= 1 - {1 \over 2}\left({k \over {k+1}}\right)\left(1 - {1 \over k}\right) 
     = 1 - {1 \over 2}\left({{k - 1} \over {k+1}}\right)
  \end{align}

 Now, let us compute, $\mathbb{E}\left[{\mu(q) \over \mu(p)}\ \mid\ \mu(p) > \mu(q)\right]$. Recall that $\mu(p) = max_{1 \le i \le k}{p[j]}$. So the distribution function of $\mu(p)$ is,
 \begin{align}
   \nonumber F[\mu(p)] = \operatorname{Pr}[\mu(p) < t] = \operatorname{Pr}\left[\bigwedge_{1\le i\le k}{p[i] < t}\right] = \prod_{1\le i\le k}\operatorname{Pr}[p[i] < t] = t^k
 \end{align} 
 Where the second equality comes from that fact that each component of $p$ are independent and identically distributed on $[0, 1]$ with uniform probability. Hence,
  \begin{align}
     \nonumber \mathbb{E}\left[{\mu(q) \over \mu(p)}\ \mid\ \mu(p) > \mu(q)\right]   &=   \int_{\mu(p) > \mu(q)}{{\mu(q) \over \mu(p)}}dF[\mu(p)]dF[\mu(q)] \\ \nonumber &=  {{k^2} \over {\operatorname{Pr}[\mu(p) > \mu(q)]}}\int_{0}^{1}{\int_{0}^{\mu(p)}{\mu(p)^{k-2}\mu(q)^{k}}d\mu(p)d\mu(p)} \\ \nonumber  &= {k \over {k+1}}
  \end{align}  
   A similar argument can be used to prove the second claim.
\end{proof}

\subsection*{Solving $a(w,d)$}
To simplify our calculations we modify the recurrence slightly: With $a(w,d) = k^d b(w, d)$, the recurrence equation becomes,

\begin{align}
 \nonumber  b(w,d) = {1 \over k^d}b(w-1,d-1) + {(1 - {1 \over k^d})}b(w-1,d)
\end{align}
Let, $G_d(z) = \sum_{w = 0}^{\infty}{b(w,d)z^w}$. We note that $b(w, d) = 0$ when $w \le d$. Then we have,
\begin{align}
  \nonumber  G_d(z) &= {z \over k^d}G_{d-1}(z) + z(1 - {1 \over k^d})G_d(z) \\
   \nonumber & = {z \over {k^d (1 - (1 - {1 \over k^d})z)}}G_{d-1}(z) \\
  \nonumber & ... \\
  \nonumber  & = {{z^d} \over {\prod_{i = 1}^{d}{k^i(1 - (1 - {1 \over k^i})z)}}}G_0(z)
\end{align}
But, $G_0(z) = \sum_{w = 0}^{\infty}{b(w, 0)z^w} = \sum_{w = 1}^{\infty}{z^w} = {z \over {1 - z}}$ as $b(w, 0) = a(w, 0)  = 1$ when $w \ge 1$. Hence,

\begin{align}
  \nonumber b(w, d) &= k^{-d(d + 1)/2}[z^{w - d - 1}] G_d(z) \\
   &= k^{-d(d + 1)/2}[z^{w - d - 1}]{{1} \over {(1 - z)\prod_{i = 1}^{d}{(1 - (1 - k^{-i})z)}}}
\end{align}
Where the notation $[z^i]p(z)$ means the coefficient of $z^i$ in the polynomial $p(z)$ as usual. Using partial fractions: Let,

\nonumber
\begin{align}
   {{1} \over {(1 - z)\prod_{i = 1}^{d}{(1 - (1 - k^{-1})z)}}} \equiv {\beta_0 \over {1 - z}} + \sum_{i = 1}^{d}{\beta_i \over {(1 - (1 - k^{-i})z)}}
\end{align} 
For which we get the following solution,
\begin{align}
   &\beta_0 = k^{d(d + 1)/2} \\
   &\beta_i = {{k^{d(d + 1)/2}(1 - k^{-i})^d} \over {\prod_{j \ne i, j \ge 1}^{d}{(1 - k^{j - i})}}}
 \end{align}
Substituting these in Equation 4 above we get, $b(w,d) = 1 - \sum_{i = 1}^{d}{{(1 - k^{-i})^{w - 1}}\over{\prod_{j = 1, j \ne i}^{d}{(1 - k^{j - i})}}}$, which gives us the desired result for $a(w,d)$.

\subsection*{Proof of Lemma 4}
We have, $S = \sum_{i = 0}^{n}{b_im^i}$, where $b_r \ge b_{r+1} \ge ... \ge b_n$. But then, 
\nonumber
\begin{align}
 S & = \sum_{i = 0}^{r}{b_im^i} + \sum_{i = r + 1}^{n}b_{i}m^{i}
  \le \sum_{i = 0}^{r}{b_im^i} + \sum_{i = r + 1}^{n}b_{r+1}m^{i} 
  \nonumber\\ &= \sum_{i = 0}^{r}{b_im^i} + {b_{r+1}m^{r+1} \over {1 - m}}
\end{align}

\subsection*{Computing $d_0$}
We shall denote $d_0$ as $d_0(w,k)$ as it is a function of both $w$ and $k$. Hence, $d_0(w,k)$ maximizes $a(w,d)$ as $d$ varies from $0$ to $w - 1$. Since we are interested in an upper bound on $d_0(w,k)$, we may think of $d$ being fixed and we vary $w$ from 0 to $\infty$\footnote{actually from $d+1$ as terms below it are 0, but this does not affect our analysis}. If we then lower bound $w$, when $a(w,d)$ maximum, we will be able get a corresponding upper bound for $d_0(w,k)$. This makes our analysis simpler as the number of terms in the expression for $a(w,d)$ is fixed for a fixed d. 

Let,
\nonumber
\begin{align}
 a'(w,d) & = a(w, d-1) - a(w, d) \nonumber \\ &= k^{d-1}\left(1 - \sum_{i = 1}^{d-1}{{(1 - {1\over {k^i}})^{w - 1}}\over{\prod_{j = 1, j \ne i}^{d-1}{(1 - {1\over {k^{i - j}}})}}}\right) - k^{d}\left(1 - \sum_{i = 1}^{d}{{(1 - {1\over {k^i}})^{w - 1}}\over{\prod_{j = 1, j \ne i}^{d}{(1 - {1\over {k^{i - j}}})}}}\right)  
\end{align}
 Letting $\alpha_i = 1 - {1 \over k^i}$ we get,
 \nonumber
 \begin{align}
 a'(w,d) &= -k^{d - 1}(k - 1) + k^{d-1}\sum_{i = 1}^{d}{{{\alpha_i}^{w-1}(k - \alpha_{i - d})} \over {\prod_{j = 1, j \ne i}^{d}{\alpha_{i-j}}}}
\end{align}
Since we wish to compute $d_0(w,k)$ or at least get an upper bound, we assume that $a'(w,d) < 0$. Hence,
 \begin{align}
 \sum_{i = 1}^{d}{{{\alpha_i}^{w-1}(k - \alpha_{i - d})} \over {\prod_{j = 1, j \ne i}^{d}{\alpha_{i-j}}}} < k - 1 
\end{align}
Since, $\prod_{j = 1, j \ne i}^{d}{\alpha_{i-j}} = \prod_{j = 1}^{i - 1}{\alpha_{i - j}}\prod_{j = 1}^{d-i}{(1 - k^j)} = P(i-1)\prod_{j = 1}^{d-i}{(1 - k^j)}$, where, $P(i) = \prod_{j = 1}^{i-1}{\alpha_{i - j}}$.

Let $A = k \sum_{i = 1}^{d}{{{\alpha_i}^{w-1}} \over {P(i-1)\prod_{j = 1}^{d-i}{1 - k^j}}}$ and $B = \sum_{i = 1}^{d}{{{\alpha_i}^{w-1}\alpha_{i-d}} \over {P(i-1)\prod_{j = 1}^{d-i}{1 - k^j}}}$. Then according to our assumption, $A + B < k - 1$.
\\
However, writing out the terms in the expression for $A$ yields:

\begin{align}
 A &= k\left({ {\alpha_{d}^{w-1}} \over {P(d - 1)}} - { {\alpha_{d-1}^{w-1}} \over {(k - 1)P(d - 2)}} + { {\alpha_{d-2}^{w-1}} \over {(k - 1)(k^2 - 1)P(d - 3)}} - ...\right)\\
 &= k\left({ {\alpha_{d}^{w-1}} \over {P(d - 1)}} - { {\alpha_{d-1}^{w-1}} \over {(k - 1)P(d - 2)}} + o({1 \over k^2}) - ...\right)\\
 & \ge k\left({ {\alpha_{d}^{w-1}} \over {P(d - 1)}} - { {\alpha_{d-1}^{w-1}} \over {(k - 1)P(d - 2)}}\right)
\end{align}
It is not difficult to see that  ${ {\alpha_{d-2}^{w-1}} \over {(k - 1)(k^2 - 1)P(d - 3)}} - ... \le o({1 \over k^2})$. Since, $P(i) < P(j)$ when $i > j$ and $\alpha_i^{w - 1}$ decreases as $i \to 0$.
Similarly, we can show that,
\begin{align}
   B \ge - {\alpha_{d-1}^{w-1} \over P(d-2)}
\end{align}
Thus we get,
\begin{align}
   k{ {\alpha_{d}^{w-1}} \over {P(d - 1)}} - { {\alpha_{d-1}^{w-1}} \over {P(d - 2)}}\left(1 + {k \over { k-1}}\right) < k - 1
\end{align}
Since we want to get an upper bound for $d_0(w,k)$ we assume that $w$ is sufficiently large. More precisely, we let $w \ge ck^{d-1} + 1$ where $0 < c < 1$. But then,
\begin{align}
    \alpha_{d-1}^{w-1} &= \left(1 - {1 \over{ k^{d-1}}}\right)^{w - 1} \le \left(1 - {1 \over{ k^{d-1}}}\right)^{ck^{d-1}}\\
    & \le e^{-c} \approx {3 \over 5}
\end{align}
Here we take $c = {1 \over 2}$. Putting this value for $\alpha_{d-1}^{w-1}$ in the main equation  and dividing both sides by $k$ we get,
\begin{align}
   { {\alpha_{d}^{w-1}} \over {P(d - 1)}} < 1 - {1 \over k} + {3 \over 5}{{2k - 1} \over {k(k - 1)P(d - 1)}}
\end{align}
Since $P(d - 1) < P(d - 2)$. But, 
\begin{align}
   P(d - 1) &> P(d) =\left (1 - {1 \over k}\right)\left(1 - {1 \over k^2}\right)...\left(1 - {1 \over k^{d-1}}\right)\\
   & = 1 - {1 \over k} - {1 \over k^2} + {1 \over k^5} + {1 \over k^7} + ... \ge 1 - {1 \over k} - {1 \over k^2}
\end{align}
And $P(d - 1) <  1 - {1 \over k}$. Substituting these upper and lower bound of $P(d - 1)$ to LHS and RHS of the expression respectively we get,
\begin{align}
   k{\alpha_{d}^{w-1} \over {k-1}} < 1 - {1 \over k} + {3 \over 5}{{(2k - 1)k} \over {(k - 1)(k^2 - k -1)}} \le 1 + {\gamma(k) \over k}
\end{align}
Where $0 <\gamma(k) < 1$ for $k \ge 4$. So we have,
\begin{align}
   1 - {w \over k^d} < \left(1 - {1 \over k^d}\right)^{w} < \alpha_{d}^{w-1} < \left(1 + {\gamma(k) \over k}\right)\left(1 - {1 \over k}\right) < 1 - {1 \over k^2}
\end{align}
So we get, $d < \log_{k}{w} + 2$. This is the upper bound on $d_0(w,k)$ that we have sought for.

\subsection *{Derivation of Equation 3 in Theorem 3}
We know $d_0 \le \beta + 3$. Where, $\beta = \lfloor \log_k{w} \rfloor$. Hence (for $k \ge 4$),

\begin{align}
   u(w) &= \sum_{d = 0}^{w - 1}{\eta_1^d a(w, d)} \nonumber \\
    &\le \sum_{d = 0}^{d_0}{\eta_1^d a(w, d)} + {7 \over 3}\eta_1^{d_0}a(w, d_0 + 1) \nonumber \\
    &\le \sum_{d = 0}^{\beta}{\eta_1^d k^d} +   {\eta_1}^{\beta + 1}(a(w, \beta + 1) + a(w, \beta + 2) + a(w, \beta + 3)) + {7 \over 3}\eta_1^{\beta + 3}a(w, \beta + 4)
\end{align}
Here we observe that, $a(w,\beta + i) \le w$ for $1 \ge i \ge 4$. Substituting these bounds we get,

\begin{align}
   u(w) &\le {{(\eta_1 k)^{\beta + 1} - 1} \over {\eta_1 k - 1}} + c_1w\eta_1^{\beta + 1} 
\end{align}
Were $c_1 < 5$ is a constant. Thus,
\begin{align}
   u(w) &\le {{(\eta_1 k)^{\beta}(\eta_1 k)} \over {\eta_1 k - 1}} + c_1w\eta_1^{\beta + 1} \\
   & \le c_2(\eta_1 k)^{\beta} + c_1w\eta_1^{\beta + 1}\\
   & \le c_2(\eta_1 k)^{\lfloor \log_k{w} \rfloor} + c_1w\eta_1^{\lfloor \log_k{w} \rfloor + 1}\\
   & \le c_2(\eta_1 k)^{\log_k{w}} + c_1w\eta_1^{\log_k{w}} \\
   & \le (c_1 + c_2)w\eta_1^{\log_k w} \le c_3w^{1 - {1 \over {\log k}} + \log_k{(1 + {2 \over {k+1}})}}
\end{align}
Here, $c_2 \le {14 \over 9}$ and $c_3$ are constants. 
\end{document}